\documentclass[times,12pt]{article}
\usepackage{latex8,times,amsmath,amssymb,amsthm}
\newcommand{\namedref}[2]{\hyperref[#2]{#1~\ref*{#2}}}

\newcommand{\CC}{\mathcal{C}}

\newcommand{\F}{\mathbb{F}}

\theoremstyle{plain}

\newtheorem{theorem}{Theorem}
\newtheorem{lemma}[theorem]{Lemma}
\newtheorem{proposition}[theorem]{Proposition}

\newtheorem{definition}[theorem]{Definition}

\newtheorem{exmp}[theorem]{Example}
\newenvironment{example}{\begin{exmp}
\begin{normalfont}}{\end{normalfont}
\end{exmp}}

\theoremstyle{definition}
\theoremstyle{plain}

\newcommand{\eat}[1]{}

\renewcommand{\geq}{\geqslant}
\renewcommand{\leq}{\leqslant}

\renewcommand{\epsilon}{\varepsilon}

\def\bull{\vrule height .9ex width .8ex depth -.1ex }

\title{Explicit Maximally Recoverable Codes with Locality}
\author{Parikshit Gopalan \\ Microsoft Research \\ parik@microsoft.com
\and Cheng Huang \\ Microsoft Research \\ chengh@microsoft.com
\and Bob Jenkins \\ Microsoft Corporation \\ bob.jenkins@microsoft.com
\and Sergey Yekhanin \\ Microsoft Research \\ yekhanin@microsoft.com}

\begin{document}
\date{}
\maketitle

\begin{abstract}
Consider a systematic linear code where some (local) parity symbols depend on few prescribed symbols, while other (heavy) parity symbols may depend on all data symbols. Local parities allow to quickly recover any single symbol when it is erased, while heavy parities provide tolerance to a large number of simultaneous erasures. A code as above is maximally-recoverable, if it corrects all erasure patterns which are information theoretically recoverable given the code topology. In this paper we present explicit families of maximally-recoverable codes with locality. We also initiate the study of the trade-off between maximal recoverability and alphabet size.
\end{abstract}

\section{Introduction}\label{Sec:Intro}

We say that a certain coordinate of an error-correcting code has locality $r$ if, when erased, the value at this coordinate can be recovered by accessing at most $r$ other coordinates. Recently there has been two lines of work on codes with locality.

In~\cite{GHSY} motivated by applications to distributed storage~\cite{HuangSX} the authors studied systematic linear $[n,k]$ codes that tolerate up to $h+1$ erasures, but also have locality $r$ for all information coordinates. In canonical codes of~\cite{GHSY}, $r$ divides $k$ and $n=\frac{k}{r}+h.$ Data symbols are partitioned into $\frac{k}{r}$ groups of size $r.$ For each data group there is a local parity storing the XOR of respective data symbols. There also are $h$ heavy parities where each heavy parity depends on all $k$ data symbols. In what follows we refer to codes above as {\it data-local} $(k,r,h)$-codes.

In~\cite{BHH} motivated by applications to data storage on SSDs the authors studied systematic linear $[n,k]$ codes with two extra parameters $r$ and $h,$ where $r\mid(k+h).$ In codes of~\cite{BHH} there are $k$ data symbols and $h$ heavy parity symbols. These $(k+h)$ symbols  are partitioned into $\frac{k+h}{r}$ groups of size $r.$ For each group there is a local parity storing the XOR of respective symbols. Thus $n=k+h+\frac{k+h}{r}.$ Unlike the codes of~\cite{GHSY}, codes of~\cite{BHH} provide locality for all symbols data or parity. In what follows we refer to codes above as {\it local} $(k,r,h)$-codes.

Observe that our descriptions of code families above are so far incomplete. For every parity symbol we specified other symbols that it depends on, i.e., we have fixed codes' topology. To complete defining the codes we need to set coefficients in the heavy parity symbols. Different choices of coefficients lead to codes with different erasure correcting capabilities. Ideally, we would like our codes to correct all patterns of erasures that are correctable for some setting of coefficients in heavy parities. Such codes exist and are called Maximally Recoverable (MR)~\cite{CHL}.

An important problem left open by earlier work has been to come up with explicit maximally-recoverable data-local and local codes over small finite fields.

\subsection{Our results and related work}

In this paper we make progress on the problem above. We present the first explicit families of maximally-recoverable data-local and local codes for all values of $k,r$ and $h.$ Prior to our work infinite explicit families of maximally-recoverable local codes were known only for $h=1$ and $h=2.$ There have also been few constructions that involved computer search for coefficients~\cite{BHH,Blaum}. Our codes improve upon the earlier constructions both in concrete settings and asymptotically.

In the asymptotic setting of $h=O(1), r=O(1),$ and growing $k$ our codes use alphabet of size $O\left(k^{h-1}\right).$ In the case of $h\geq 2^r$ the alphabet size can be reduced $O\left(k^{\left\lceil (h-1)\left(1-\frac{1}{2^r}\right)\right\rceil}\right).$ We also obtain further improvements in the special cases of $h=3$ and $h=4.$ The only lower bound for the alphabet size known currently comes from results on the main conjecture for MDS codes~\cite{MS} and is $\Omega(k).$ One way to construct maximally-recoverable local codes is by picking coefficients in heavy parities at random from a large enough finite field. In order to compare our constructions with random codes we show that random codes are not maximally recoverable (except with probability $o(1)$) unless the size of the finite field from which the coefficients are drawn exceeds $\Omega\left(k^{h-1}\right).$

Similarly to~\cite{BHH,Blaum} we construct our explicit codes via parity check matrices. As in~\cite{BHH} columns of our parity check matrices have the shape $(\alpha_i,\alpha_i^2,\ldots,\alpha_i^{2^{h-1}}).$ The key difference from the work of~\cite{BHH,Blaum} however is that we explicitly specify the sets $\{\alpha_i\}$ used in our constructions.

There are several other models of codes with locality in the literature. The ones most closely related to our work include SD codes~\cite{Blaum,SD_fail}, locally decodable codes~\cite{Y_now}, and regenerating codes~\cite{Dimakis_1}.

\subsection{Organization}
In section~\ref{Sec:Prelim} we formally define data-local and local $(k,r,h)$-codes. We introduce the notion of maximal recoverability, and show that maximally-recoverable local codes yield maximally-recoverable data-local codes. In section~\ref{Sec:Constructions} we give our two main code constructions. In section~\ref{Sec:Asymptotics} we analyze the asymptotic behavior of alphabet size in our codes for large message lengths. We also establish a simple lower bound on the alphabet size of maximally-recoverable local codes. Finally, we compare asymptotic parameters of our codes to asymptotic parameters of random codes. In section~\ref{Sec:OpenQ} we conclude with open questions.

\section{Preliminaries}\label{Sec:Prelim}

We use the following notation
\begin{itemize}

\item For an integer $n,$ $[n]=\{1,\ldots,n\};$

\item An $[n,k]$ code is a linear code encoding $k$-dimensional messages to $n$-dimensional codewords. Equivalently, one can think of an $[n,k]$ code as a $k$-dimensional subspace of an $n$-dimensional space over a finite field;

\item An $[n,k,d]$ code is an $[n,k]$ code whose minimal distance is at least $d;$

\item Let $C$ be an $[n,k]$ code and $S\subseteq [n].$ Puncturing $C$ in coordinates in $S$ means restricting $C$ to coordinates in $[n]\setminus S.$ It yields a $\left[k^\prime,n-|S|\right]$ code $C^\prime,$ where $k^\prime\leq k.$
\end{itemize}

We proceed to formally introduce the notion of locality~\cite{GHSY}.

\begin{definition}\label{Def:Locality}
Let $C$ be a linear $[n,k]$ code. We say that the $i$-th coordinate of $C$ has locality $r,$ if there exists a set $S\subseteq [n]\setminus\{i\},$ $|S|\leq r,$ such that across all codewords ${\bf c}\in C,$ the value of the coordinate ${\bf c}(i)$ is determined by values of coordinates $\{{\bf c}(j)\}, j\in S.$ Equivalently, the $i$-th coordinate has locality $r,$ if the dual code $C^{\perp}$ contains a codeword ${\bf c}$ of Hamming weight at most $r+1,$ where coordinate $i$ is in the support of ${\bf c}.$
\end{definition}

\begin{definition}\label{Def:DLcode}
Let $C$ be a linear systematic $[n,k]$ code. We say that $C$ is a $(k,r,h)$ data-local code if the following conditions are satisfied:
\begin{itemize}
\item $r\mid k$ and $n=\frac{k}{r}+h;$
\item Data symbols are partitioned into $\frac{k}{r}$ groups of size $r.$ For each such group there is one (local) parity symbol that stores the XOR of respective data symbols;
\item Remaining $h$ (heavy) parity symbols, may depend on all $k$ data symbols.
\end{itemize}
\end{definition}
In what follows we refer to a group of $r$ data symbols and their local parity defined above as a {\it local group}. Data-local codes have been studied in~\cite{HuangCL07,GHSY,PD_isit,PKLK,XOR_ELE,FY}. The importance of this topology was partially explained in~\cite[Theorem 9]{GHSY}.  There it has been shown that in case $h<r+1,$ any systematic $[n,k]$ code that corrects all patterns of $(h+1)$ erasures, provides locality $r$ for all data symbols, and has the lowest possible redundancy has to be a data-local $(k,r,h)$-code. The class of data local-codes is fairly broad as there is a lot of flexibility is choosing coefficients in heavy parities. Below we define data-local codes that maximize reliability.

\begin{definition}\label{Def:MRDLcode}
Let $C$ be a data-local $(k,r,h)$-code. We say that $C$ is maximally-recoverable if for any set $E\subseteq [n],$ where $E$ is obtained by picking one coordinate from each of $\frac{k}{r}$ local groups, puncturing $C$ in coordinates in $E$ yields a maximum distance separable $[k+h,k]$ code.
\end{definition}

A $[k+h,k]$ MDS code obviously corrects all patterns of $h$ erasures. Therefore a maximally-recoverable data-local $(k,r,h)$-code corrects all erasure patterns $E\subseteq [n]$ that involve erasing one coordinate per local group, and $h$ additional coordinates. We now argue that any erasure pattern  that is not dominated by a pattern above has to be uncorrectable.

\begin{lemma}\label{Lemma:DLuncorrectable}
Let $C$ be an arbitrary data-local $(k,r,h)$-code. Let $E\subseteq [n]$ be an erasure pattern. Suppose $E$ affects $t$ local groups and $|E|>t+h;$ then $E$ is not correctable.
\end{lemma}
\begin{proof}
Suppose $E$ is correctable. We extend $E$ to a larger pattern of erasures $E^\prime$ erasing one arbitrary coordinate in each of $\frac{k}{r}-t$ local groups that are not affected by $E.$ Observe that $E^\prime$ is correctable if $E$ is correctable since each local group has a local parity. Note that the size of $E^\prime$ exceeds redundancy of the code $C,$ $\left|E^\prime\right|>\frac{k}{r}+h.$ Thus the dimension of $C$ restricted to coordinates outside of $E^\prime$ is below $k,$ and there are codewords in $C$ with identical projections on $[n]\setminus E^\prime.$ Therefore $E^\prime$ is not correctable.
\end{proof}

We now proceed to define local codes.

\begin{definition}\label{Def:LCode}
Let $C$ be a linear systematic $[n,k]$ code. We say that $C$ is a $(k,r,h)$ local code if the following conditions are satisfied:
\begin{itemize}
\item $r\mid (k+h)$ and $n=k+h+\frac{k+h}{r};$
\item There are $k$ data symbols and $h$ heavy parity symbols, where each heavy parity may depend on all data symbols;
\item These $k+h$ symbols are partitioned into $\frac{k+h}{r}$ groups of size $r.$ For each such group there is one (local) parity symbol that stores the XOR of respective symbols.
\end{itemize}
\end{definition}

We refer to a group of $r$ symbols and their local parity as a {\it local group}. As above we now introduce local codes that maximize reliability.

\begin{definition}\label{Def:MRLcode}
Let $C$ be a local $(k,r,h)$-code. We say that $C$ is maximally-recoverable if for any set $E\subseteq [n],$ where $E$ is obtained by picking one coordinate from each of $\frac{k+h}{r}$ local groups, puncturing $C$ in coordinates in $E$ yields a maximum distance separable $[k+h,k]$ code.
\end{definition}

Maximally recoverable local $(k,r,h)$-codes have been originally introduced in~\cite{BHH} under the name of partial-MDS codes. Similarly to the discussion following definition~\ref{Def:MRDLcode} it is easy to see that these codes correct all erasure patterns that involve erasing one coordinate per local group, and $h$ additional coordinates. Erasure patterns that are not dominated by such patterns are not correctable by any local $(k,r,h)$-code. The next lemma gives a simple reduction from local MR codes to data-local MR codes.

\begin{lemma}\label{Lemma:LtoDL}
Suppose there exists a local maximally-recoverable $(k,r,h)$-code $C$ over a finite field $\mathbb{F};$ then there exists a data-local maximally-recoverable $(k^\prime,r,h)$-code $C^\prime$ over the same field, where $k^\prime\leq k$ is the largest integer that is divisible by $r.$
\end{lemma}
\begin{proof}
Let $t=\frac{k+h}{r}.$ Let $G_1,\ldots,G_t\subseteq [n]$ be the local groups. $\cup_i G_i=[n].$ We refer to data symbols and heavy parity symbols of $C$ as {\it primary} symbols. Altogether primary symbols form a $[k+h,k]$ MDS code. Note that any $k$ symbols of an MDS code can be treated as information symbols. Next we consider two cases:
\begin{itemize}
\item $r\mid k.$ We treat $k$ primary symbols of $C$ that belong to local groups $\{G_i\},i\leq \frac{k}{r}$ as data symbols of $C^\prime.$ The code $C^\prime$ is obtained form the code $C$ by dropping local parity symbols from groups $G_i$ for $i> \frac{k}{r}.$ The code $C^\prime$ clearly satisfies definition~\ref{Def:DLcode}. Observe that $C^\prime$ also satisfies definition~\ref{Def:MRDLcode} as any code that can be obtained by dropping one coordinate per local group in $C^\prime$ can also be obtained by dropping one coordinate per local group in $C.$

\item $r\nmid k.$ Let $s=\lfloor\frac{k}{r}\rfloor.$ We refer to local groups $\{G_i\}, i\leq s$ as data groups. We refer to group $G_{s+1}$ as special. We treat $k^\prime$ primary symbols of $C$ that belong to data groups $\{G_i\},i\leq s$ as data symbols of $C^\prime.$ We fix some arbitrary $k-k^\prime$ primary symbols in the special group, and refer to them as special symbols. We denote the collection of special symbols by $S.$

    The code $C^\prime$ is obtained from the code $C$ by dropping all special symbols and $t-s$ local parities in groups other than data groups. Given an assignment of values to $k^\prime$ data symbols of $C^\prime$, we determine the values of heavy parities using the code $C$ assuming that all special symbols are set to zero.

    The code $C^\prime$ clearly satisfies definition~\ref{Def:DLcode}. It also satisfies definition~\ref{Def:MRDLcode} as any codeword that can be obtained by dropping one coordinate per local group in $C^\prime(x^\prime)$ can also be obtained by dropping one coordinate per local group in $C(x^\prime\circ 0^{k-k^\prime})$ restricted to $[n]\setminus S.$ The latter restriction does not affect the erasure correcting capability of the code as we are dropping coordinates that are identically zero.
\end{itemize}

This concludes the proof.
\end{proof}

\section{Code constructions}\label{Sec:Constructions}

In this section we give our two main constructions of local codes. We restrict our attention to finite fields of characteristic two. Let $\F$ be such a field. Let $S=\{\alpha_1,\ldots,\alpha_n\} \subseteq \F$ be a multi-set of $n$ elements. Let $A(S,h) = [a_{ij}]$ denote the $h \times n$ matrix where
\[ a_{ij} = \alpha_j^{2^{i-1}} \]
Let $\CC(\alpha,h) \subset \F^n$ be the linear code whose parity check matrix is $A$. Equivalently, $\CC(\alpha,h)$ contains all vectors ${\bf x} =
(x_1,\ldots,x_n)$ which satisfy the equations
\begin{equation}\label{Eqn:Frob}
\sum_{i=1}^n \alpha_i^{2^{j-1}}x_i = 0 \ \ \text{for}\ \ j = 1,\ldots,h.
\end{equation}

Let $\CC(\alpha,h)$ be an $[n,k,d]$ code. It is easy to see that $k \geq n -h$, hence by the Singleton bound, $d \leq h+1$. We are interested in sets $\{\alpha_i\}$ where $d = h+1$, so that the code $\CC(\alpha,h)$ is maximum distance separable. The following lemma characterizes such sets.

\begin{definition}\label{Def:KwiseIndep}
We say that the multi-set $S \subseteq \F$ is $t$-wise independent over a field $\F' \subseteq \F$ if every $T \subseteq S$ such that $|T| \leq t$ is linearly independent over $\F'$.
\end{definition}

\begin{lemma}
\label{lem:mds}
The code $\CC(S,h)$ has distance $h+1$ if and only if the multi-set $S$ is $h$-wise independent over the field $\F_2.$
\end{lemma}
\begin{proof}
Let ${\bf x} = (x_1,\ldots,x_n) \in \CC(S,h)$ be a codeword. The code $\CC(S,h)$ has distance $h+1$ iff every pattern of $h$ erasures is correctable. In other words, for any $E \subseteq [n]$, the values $\{x_i\}_{i\in E}$ can be recovered if we know the values of all $\{x_i\}_{i \in [n]\setminus E}$.  This requires solving the following system of equations:
\begin{equation}
\sum_{i \in E}\alpha_i^{2^{j-1}}x_i  = b_j, \ \ 1 \leq j \leq h
\end{equation}
which in turn requires inverting the $h \times h$ matrix $A_E$ which is the the minor of $A$ obtained by taking the columns in $E$. It is easy to see (e.g.,~\cite[Lemma 3.51]{LN}) that $A_E$ is invertible if and only if the multi-set $\{\alpha_i\}_{i\in E}$ is linearly independent over $\F_2.$
\end{proof}

Lemma~\ref{lem:mds} describes the effect of adding parity check constraints to $n$ otherwise independent variables. We now consider the effect of adding such constraints to symbols that already satisfy some dependencies. We work with the following setup. The $n$ coordinates of the code are partitioned into $\ell=\frac{n}{r+1}$ local groups, with group $i$ containing $r+1$ symbols $x_{i,1},\ldots,x_{i,r+1}$. Variables in each local group satisfies a parity check constraint $\sum_{s=1}^{r+1} x_{i,s} = 0$. Thus all code coordinates have locality $r.$ Let
$$
S = \{\alpha_{i,s}\}_{i \in [\ell], s \in [r+1]} \in \F^n
$$
We define the code $\CC(S,r,h)$ by the parity check equations
\begin{eqnarray}
\sum_{i=1}^{\ell}\sum_{s=1}^{r+1} \alpha_{i,s}^{2^{j-1}}x_{i,s} = 0 & \text{for} \ j \in \{1,\ldots,h\},\label{eq:global}\\
\sum_{s=1}^{r+1} x_{i,s} = 0 & \text{for} \ i \in \{1,\ldots,\ell\}\label{eq:local}
\end{eqnarray}
We refer to Equations~\eqref{eq:global} as global constraints and~\eqref{eq:local} as local constraints. The following proposition is central to our method:
\begin{proposition}\label{Prop:Central}
Let $\CC(S,r,h)$ be the code defined above. Let ${\bf e}\in [r+1]^\ell$ be a vector. Let $\CC^{-{\bf e}}=\CC^{-{\bf e}}(S,r,h)$ be the code obtained by puncturing $\CC(S,r,h)$ in positions $\{i,{\bf e}(i)\}_{i =1}^\ell$. Then $\CC^{-{\bf e}}$ is an MDS code if and only if the multi-set
\begin{align*}
T(S,{\bf e}) = \{\alpha_{i,s} + \alpha_{i,{\bf e}(i)}\}_{i \in [\ell], s \in [r+1]\setminus \{{\bf e}(i)\}}
\end{align*}
is $h$-wise independent.
\end{proposition}
\begin{proof}
Note that $\CC^{-{\bf e}}$ is a $[k+h,k]$ code. To prove that it is MDS, we will use the local parity constraints to eliminate the punctured locations and then use Lemma~\ref{lem:mds}. Firstly, by renumbering variables (and coefficients $\{\alpha_{i,s}\}$) in each local group, we may assume ${\bf e}(i) = r+1.$ By the local parity check equations,
$$ x_{i,r+1} =  \sum_{s=1}^{r}x_{i,s}. $$
We use these to eliminate $x_{i,r+1}$ from the global parity check equations for $j \in [h]:$
\begin{align*}
0 & = \sum_{i=1}^{\ell}\left(\sum_{s=1}^{r+1} \alpha_{i,s}^{2^{j-1}}x_{i,s}\right) \\
& = \sum_{i=1}^{\ell}\left(\left(\sum_{s=1}^{r} \alpha_{i,s}^{2^{j-1}}x_{i,s}\right) + \alpha_{i,r+1}^{2^{j-1}}\left(\sum_{s=1}^{r}x_{i,s}\right)\right) \\
& = \sum_{i=1}^{\ell}\left(\sum_{s=1}^{r}
(\alpha_{i,s}^{2^{j-1}} + \alpha_{i,r+1}^{2^{j-1}})x_{i,s}\right) \\
& = \sum_{i=1}^{\ell}\left(\sum_{s=1}^{r}
(\alpha_{i,s} + \alpha_{i,r+1})^{2^{j-1}}x_{l,i}\right).
\end{align*}
Let $T =\{\alpha_{i,s} + \alpha_{i,r+1}\}_{i \in [\ell],s \in[r]}.$ By Lemma~\ref{lem:mds}, the code $\CC^{-{\bf e}}$ is MDS if and only if $T$ is $h$-wise independent.
\end{proof}

Proposition~\ref{Prop:Central} reduces constructing local MR codes to obtaining multi-sets $S\subseteq \mathbb{F}$ such that all sets $T(S,{\bf e})$ are $h$-wise independent. In what follows we give two constructions of such multi-sets.

\subsection{Basic construction}\label{Sec:BasicConstruction}

\begin{lemma}\label{Lemma:2hyeildsTse}
Let $S\subseteq \mathbb{F}, |S|=n$ be a set that is $2h$-wise independent over as subfield $\mathbb{F}^\prime.$ Let $r$ be arbitrary such that $\ell=\frac{n}{r+1}$ is an integer. Then for all ${\bf e}\in [r+1]^\ell$ the set $T(S,{\bf e})$ is $h$-wise independent over $\mathbb{F}^\prime.$
\end{lemma}
\begin{proof}
Assume the contrary. To simplify the notation we relabel variables and assume that ${\bf e}(i)=r+1$ for every $i\in [\ell].$ Let $D = \{i_j,s_j\}_{j=1}^d$ be a set of $d \leq h$ indices of $T$ such that
\begin{align*}
\sum_{j=1}^d\left(\alpha_{i_j,s_j} + \alpha_{i_j,r+1} \right)= 0
\end{align*}
We can rewrite this as
\begin{align*}
\sum_{j=1}^d\alpha_{i_j,s_j} + \sum_{j=1}^d\alpha_{i_j,r+1} = 0
\end{align*}
We claim that this gives a non-trivial relation between the coefficients $\{\alpha_{i,s}\}.$ The relation is non-trivial because the terms in the first summation occur exactly once (whereas terms in the second summation can occur multiple times depending on the set $D$ and could cancel).
\end{proof}

Observe that the task of constructing $n$-sized subsets of $\mathbb{F}_{2^t}$ that are $2h$-wise independent over $\mathbb{F}_2$ is equivalent to the task of constructing $[n,n-t,2h+1]$ binary linear codes, as elements of a $2h$-wise independent set can be used as columns of a $t\times n$ parity check matrix of such a code, and vice versa. Therefore any family of binary linear codes can be used to obtain maximally-recoverable local codes via Lemma~\ref{Lemma:2hyeildsTse} and Proposition~\ref{Prop:Central}. The next theorem gives local MR codes that one gets by instantiating the approach above with columns of the parity check matrix of a binary BCH code.

\begin{theorem}\label{Th:BasicConstruction}
Let positive integers $k,r,h$ be such that $r\mid (k+h).$ Let $m$ be the smallest integer such that $n=k+h+\frac{k+h}{r}\leq 2^m-1.$ There exists a maximally recoverable local $(k,r,h)$-code over the field $\mathbb{F}_{2^{hm}}.$
\end{theorem}
\begin{proof}
Let $S^\prime=\{\beta_1,\ldots,\beta_n\}$ be an arbitrary subset of non-zero elements of $\mathbb{F}_{2^m}.$ Consider $S=\{\alpha_1,\ldots,\alpha_n\}\subseteq \mathbb{F}_{2^{mh}}$ where for all $i\in [n],$ $\alpha_i=(\beta_i,\beta_i^3,\ldots,\beta_i^{2h-1})$ when we treat $\mathbb{F}_{2^{mh}}$ as an $h$-dimensional linear space over $\mathbb{F}_{2^m}.$ It is not hard to see that the set $S$ is $2h$-wise independent over $\mathbb{F}_2.$ Thus by Lemma~\ref{Lemma:2hyeildsTse} and Proposition~\ref{Prop:Central} the code $\CC(S,r,h)$ is a maximally recoverable local $(k,r,h)$-code.
\end{proof}

\subsection{Optimized construction}\label{Sec:OptimizedConstruction}

In the previous section we used $2h$-wise independence of the set $S$ to ensure $h$-independence of sets $T(S,{\bf e}).$ In some cases this is on overkill, and one can ensure $h$-independence of sets $T(S,{\bf e})$ more economically.
\begin{definition}\label{Def:WeakIndependence}
We say that the set $S \subseteq \F$ is $t$-wise weakly independent over $\mathbb{F}_2 \subseteq \F$ if no set $T \subseteq S$ where $2 \leq |T| \leq t$ has the sum of its elements equal to zero.
\end{definition}
Unlike independent sets, weakly independent sets may include the zero element. The following proposition presents our approach is a general form.

\begin{proposition}\label{Prop:OptimizeCore}
Let positive integers $k,r,h$ be such that $\ell =\frac{k+h}{r}$ is an integer. Suppose there exists an $(r+1)$-sized set $S_1\subseteq \mathbb{F}_{2^a}.$ If $h$ is even we require $S_1$ to be $h$-weakly independent over $\mathbb{F}_2;$ otherwise we require $S_1$ to be $(h+1)$-weakly independent over $\mathbb{F}_2.$ Further suppose that there exists an $\ell$-sized set $S_2\subseteq \mathbb{F}_{2^b}$ that is $h$-independent over $\mathbb{F}_{2^a};$ then the code $\CC(S_1\cdot S_2,r,h)$ is a maximally recoverable local $(k,r,h)$-code over the field $\mathbb{F}_{2^b}.$
\end{proposition}
\begin{proof}
Let $S_1=\{\xi_1,\ldots,\xi_{r+1}\}.$ Let $S_2=\{\lambda_1,\ldots,\lambda_\ell\}.$ For $i\in [\ell],s\in [r+1],$ we set $\alpha_{i,s}=\lambda_i\xi_s.$ By Proposition~\ref{Prop:Central} it suffices to show that for all ${\bf e}\in [r+1]^\ell,$ the set $T(S_1\cdot S_2,{\bf e})$ is $h$-independent over $\mathbb{F}_2.$ Assume the contrary. To simplify the notation we relabel variables and assume that ${\bf e}(i)=r+1$ for every $i\in [\ell].$ Let $D = \{i_j,s_j\}_{j=1}^d$ be a set of $d \leq h$ indices of $T$ such that
\begin{align*}
\sum_{j=1}^d\left(\alpha_{i_j,s_j} + \alpha_{i_j,r+1} \right)= 0
\end{align*}
We can rewrite this as
\begin{align*}
\sum_{t\in [\ell]} \lambda_t\cdot\sum_{j\ :\ i_j=t}\left(\xi_{t,s_j}+\xi_{t,r+1}\right) =0.
\end{align*}
Observe that after cancelations each non-empty inner sum above involves at least $2$ terms. When $h$ is even it involves at most $h$ terms; when $h$ is odd it involves at most $h+1$ terms. Therefore each inner sum is non-zero by the properties of the set $S_1.$ Also note that the outer sum involves at most $h$ terms $\lambda_t$ with non-zero coefficients from $\mathbb{F}_{2^a}$ and thus is also non-zero by the properties of the set $S_2.$
\end{proof}

We now instantiate Proposition~\ref{Prop:OptimizeCore} with a certain particular choice of independent sets. Our sets come from columns of a parity check matrix of an extended BCH code.
\begin{theorem}\label{Th:OptimizedConstruction}
Let positive integers $k,r,h$ be such that $\ell =\frac{k+h}{r}$ is an integer. Let $m$ be the smallest integer such that $r\mid m$ and $\ell\leq 2^m;$ then there exists a maximally recoverable local $(k,r,h)$-code over the field $\mathbb{F}_{2^t}$ for $t=r+m\left\lceil (h-1)\left(1-\frac{1}{2^r}\right)\right\rceil.$
\end{theorem}
\begin{proof}
Let $\{\xi_1,\ldots,\xi_r\}$ be an arbitrary basis of $\mathbb{F}_{2^r}$ over $\mathbb{F}_2.$ We set $\xi_{r+1}=0$ and $S_1=\{\xi_1,\ldots,\xi_{r+1}\}.$ Clearly, $S_1$ is $(h+1)$-weakly independent over $\mathbb{F}_2$ for all $h.$ Let $S_2^\prime=\{\beta_1,\ldots,\beta_\ell\}$ be an arbitrary subset of $\mathbb{F}_{2^m}.$ Consider $S_2=\{\lambda_1,\ldots,\lambda_\ell\}\subseteq \mathbb{F}_{2^{t}}$ where for all $i\in [\ell],$
\begin{equation}\label{Eqn:Lambda}
\lambda_i=(1,\beta_i,\beta_i^2,\ldots,\beta_i^{h-1})
\end{equation}
when we treat $\mathbb{F}_{2^{t}}$ as a linear space over $\mathbb{F}_{2^r}.$ The first coordinate in~(\ref{Eqn:Lambda}) is a single value in $\mathbb{F}_{2^r},$ while every other coordinate is an $\frac{m}{r}$-dimensional vector. In formula~(\ref{Eqn:Lambda}) we also omit every non-zero power $\beta_i^j$ whenever $2^r\mid j.$ We claim that the set $S_2$ is $h$-independent over $\mathbb{F}_{2^r}.$ Assume the contrary. Then for some non-empty set $S\subseteq [\ell],$ $|S|\leq h$ for all $0\leq j\leq h-1$ whenever $2^r\nmid j$ we have
\begin{equation}\label{Eqn:LambdaDep}
\sum_{i\in S} \gamma_i \lambda_i^j,
\end{equation}
where we assume $0^0=1$ and all $\{\gamma_i\}\in \mathbb{F}_{2^r}.$ By standard properties of Frobenius automorphisms~(\ref{Eqn:LambdaDep}) implies
$$
\sum_{i\in S} \gamma_i \lambda_i^j,
$$
for all $0\leq j\leq h-1$ which contradicts the proprieties of the Vandermonde determinant.
\end{proof}

\begin{example}\label{Example:60_80}
Instantiating Theorem~\ref{Th:OptimizedConstruction} with $k=60,$ $r=h=4,$ we obtain a $[80,60,7]$ maximally recoverable $(60,4,4)$ local code over the field $\mathbb{F}_{2^{16}}$. Prior to our work~\cite[Theorem 4.2]{BHH} a code with such parameters was not known to exist over any field of size below $2^{80}.$
\end{example}

In the proof of Theorem~\ref{Th:OptimizedConstruction} we set $S_1$ to be a basis of $\mathbb{F}_{2^r}$ augmented with a zero. After that we could use columns of a parity check matrix of any linear $\left[\ell,\ell-\frac{t}{r},h+1\right]$ code over $\mathbb{F}_{2^r}$ to define the set $S_2\subseteq \mathbb{F}_{2^t}$ and obtain a MR local $(k,r,h)$-code over $\mathbb{F}_{2^t}.$ While we used columns of the parity check matrix of an extended BCH code, other choices sometimes yield local MR codes over smaller alphabets.

\subsection{Further improvements for $h=3$ and $h=4$}

In this section we carry out the steps outlined above and present codes that improve upon the codes of Theorem~\ref{Th:OptimizedConstruction} for $h=3$ or $4$ and large $k.$ We replace BCH codes in the construction of Theorem~\ref{Th:OptimizedConstruction} with better codes. The codes we use are not new~\cite{Dumer,YD}.

\begin{theorem}\label{Th:H3}
Let positive integers $k,r,h=3$ be such that $\ell =\frac{k+h}{r}$ is an integer. Let $m$ be the smallest even integer such that $\ell\leq 2^{rm};$ then there exists a maximally recoverable local $(k,r,3)$-code over the field $\mathbb{F}_{2^t}$ for $t=r(\frac{3m}{2}+1).$
\end{theorem}
\begin{proof}
Let $\{\xi_1,\ldots,\xi_r\}$ be an arbitrary basis of $\mathbb{F}_{2^r}$ over $\mathbb{F}_2.$ We set $\xi_{r+1}=0$ and $S_1=\{\xi_1,\ldots,\xi_{r+1}\}.$ Clearly, $S_1$ is $(h+1)$-weakly independent over $\mathbb{F}_2$ for all $h.$ Let $S_2^\prime\subseteq \mathbb{F}_{2^r}^{\frac{3}{2}m+1}$ be an arbitrary collection of $\ell$ columns of the parity check matrix of the code $C^\prime$ from~\cite[Theorem 5]{YD}, where we set $q=2^r$ and $d=4.$ $S_2^\prime$ naturally defines a set $S_2\subseteq \mathbb{F}_{2^t}$ that is $3$-independent over $\mathbb{F}_{2^r}$.
\end{proof}
We remark that using results in~\cite{EB} one can get further small improvements upon the theorem above.

\begin{theorem}\label{Th:H4}
Let positive integers $k,r,h=4$ be such that $\ell =\frac{k+h}{r}$ is an integer. Let $m$ be the smallest integer such that $3\mid (m-1)$ and $\ell\leq 2^{r(m-1)};$ then there exists a maximally recoverable local $(k,r,4)$-code over the field $\mathbb{F}_{2^t}$ for $t=r(2m+\frac{m-1}{3}).$
\end{theorem}
\begin{proof}
As before let $\{\xi_1,\ldots,\xi_r\}$ be an arbitrary basis of $\mathbb{F}_{2^r}$ over $\mathbb{F}_2.$ We set $\xi_{r+1}=0$ and $S_1=\{\xi_1,\ldots,\xi_{r+1}\}.$ Clearly, $S_1$ is $(h+1)$-weakly independent over $\mathbb{F}_2$ for all $h.$ Let $S_2^\prime\subseteq \mathbb{F}_{2^r}^{2m+\frac{m-1}{3}}$ be an arbitrary collection of $\ell$ columns of the parity check matrix of the code $U^\prime$ from~\cite[Theorem 6]{Dumer}, where we set $q=2^r.$ $S_2^\prime$ naturally defines a set $S_2\subseteq \mathbb{F}_{2^t}$ that is $4$-independent over $\mathbb{F}_{2^r}$.
\end{proof}

\section{Asymptotic parameters}\label{Sec:Asymptotics}

Unlike data transmission applications, in data storage applications one typically does not need to scale the number of heavy parities linearly with the number of data fragments $k$ to ensure the same level of reliability~\cite{HuangSX}, as the likelihood $p$ a fragment failure during a certain period of time is usually much smaller than $\frac{1}{k}.$ Much slower growth in the number of heavy parities suffices. Therefore we find the asymptotic setting of fixed $r,h$ and growing $k$ relevant for practice and analyze the growth rate of the alphabet size in different families of local MR $(k,r,h)$-codes in this regime.

It is not hard to see that in codes of Theorem~\ref{Th:BasicConstruction} the alphabet size grows as $O\left(k^h\right).$ In codes of Theorem~\ref{Th:OptimizedConstruction} the alphabet size grows as $O\left(k^{\lceil (h-1)\left(1-\frac{1}{2^r}\right)\rceil}\right).$ For small values of $h$ one can get some further improvements. MR local $(k,r,h=3)$-codes of Theorem~\ref{Th:H3} use alphabet of size $O\left(k^{\frac{3}{2}}\right).$ MR local $(k,r,h=4)$-codes of Theorem~\ref{Th:H4} use alphabet of size $O\left(k^{\frac{7}{3}}\right).$

Obtaining constructions with reduced alphabet size remains a major challenge. The only lower bound we currently have comes from results on the main conjecture for MDS codes and is $\Omega(k).$ In particular the asymptotic lower bound does not depend on $h.$
\begin{theorem}\label{Th:Lower}
Let $C$ be a maximally recoverable local $(k,r,h)$-code with $h\geq 2.$ Assume $C$ is defined over the finite field $\mathbb{F}_q;$ then $q\geq k+1.$
\end{theorem}
\begin{proof}
Consider the code $C^\prime$ that is obtained from $C$ by deleting all local parities. Clearly, $C^\prime$ is a $[k+h,k,h+1]$ MDS code. Consider the $h\times (k+h)$ parity check matrix of the code $C^\prime$ with entries in $\mathbb{F}_q.$ By~\cite[Lemma 1.2]{MainMDS1}, $k+h\leq q+h-1.$
\end{proof}
Details regarding the recent progress on the main conjecture for MDS codes can be found in~\cite{MainMDS1,MainMDS2}. In particular, results there allow one to get small non-asymptotic improvements upon Theorem~\ref{Th:Lower}.

\subsection{Random codes}
One way to construct maximally-recoverable local codes is by picking coefficients in heavy parities at random from a large enough finite field. In order to compare our constructions in Section~\ref{Sec:Constructions} with random local codes in this section we show that random codes are not maximally recoverable (except with probability $o(1)$) unless the size of the finite field from which the coefficients are drawn exceeds $\Omega\left(k^{h-1}\right).$ The following theorem is due to Swastik Kopparty and Raghu Meka~\cite{MK}.

\begin{theorem}\label{Th:RandomField}
Let positive integers $k,r,h$ be such that $\ell =\frac{k+h}{r}$ is an integer. Consider a local $(k,r,h)$-code $C,$ where the coefficients in heavy parities are drawn at random uniformly and independently from a finite field $\mathbb{F}_q.$ Suppose $q\leq \left(\lfloor\frac{k}{2}\rfloor\atop h-1\right);$  then the probability that $C$ is maximally-recoverable is at most $\left(1-\frac{1}{2^he^{h-1}}\right)^{\frac{k}{2}}.$
\end{theorem}
\begin{proof} Let $t\leq \lfloor\frac{k}{2}\rfloor$ be the largest integer such that $\left(t\atop h-1\right)\leq q.$ Note that for all positive integers $x$ we have
\begin{equation}\label{Eqn:BiGrow}
\left(x+1\atop h-1\right) / \left(x\atop h-1\right) \leq \left(\frac{e(x+1)}{h-1}\right)^{h-1} / \left(\frac{x}{h-1}\right)^{h-1} \leq (2e)^{h-1}.
\end{equation}
Let $\epsilon=\frac{1}{(2e)^{h-1}}.$ By~(\ref{Eqn:BiGrow}) and the definition of $t$ we have
\begin{equation}\label{Eqn:t}
\epsilon q \leq \left(t\atop h-1\right)\leq q.
\end{equation}
Consider the $[k+h,k]$ code $C^\prime$ that is obtained from $C$ by deleting all $\ell$ local parities. Let $M$ be the $h\times (k+h)$ parity check matrix  of $C^\prime.$ Columns of $M$ that correspond to heavy parities form the $h\times h$ identity matrix. Other $k$ columns ${\bf m}_1,\ldots,{\bf m}_k$ are drawn from $\mathbb{F}_q^h$ uniformly at random. In what follows for $S\subseteq [k]$ we denote the span of vectors $\{{\bf m}_i\}_{i\in S}$ by $\mathcal L(S).$ The code $C$ is maximally recoverable only if $C^\prime$ is MDS. The code $C^\prime$ is MDS only if any $h$ vectors in $\{{\bf m}_1,\ldots,{\bf m}_t\}$ are linearly independent and for all $S\subseteq [t], |S|=h-1$ and $i\in [k]\setminus [t],$ ${\bf m}_i\not\in \mathcal{L}(S).$ In what follows we assume that any $h$ vectors in $\{{\bf m}_i\}_{i\in [t]}$ are indeed independent. Let $U\subseteq \mathbb{F}_q^h$ denote the union of $\mathcal{L}(S)$ over all $S\subseteq [t], |S|=h-1.$ By inclusion-exclusion we have
\begin{equation}\label{Eqn:CupL}
|U| \geq \left(t\atop h-1\right) q^{h-1} - \left(\left(t\atop h-1\right) \atop 2\right)q^{h-2}\geq
\left(t\atop h-1\right) q^{h-1}\left(1-\frac{\left(t\atop h-1\right)}{2q}\right).
\end{equation}
By the discussion above
\begin{align*}\label{Eqn:Pr1}
\mathrm{Pr}\left[C^\prime\ \mathrm{is\ MDS}\ \right] & \leq \prod_{i=t+1}^k\mathrm{Pr}[{\bf m}_i\not\in U] \\
                                                     & =    \left(\frac{q^h - |U|}{q^h}\right)^{k-t}       \\
                                                     & \leq \left(1- \frac{\left(t\atop h-1\right)}{q}\left(1-\frac{\left(t\atop h-1\right)}{2q}\right)\right)^{k-t} \\
                                                     & \leq \left(1- \frac{\left(t\atop h-1\right)}{2q}\right)^{k-t},
\end{align*}
where the last bound follows by using the RHS of~(\ref{Eqn:t}) inside the inner brackets. Finally, using the LHS of~(\ref{Eqn:t}) in the formula above we obtain
$$
\mathrm{Pr}\left[C\ \mathrm{is\ MR}\ \right] \leq \left(1-\frac{\epsilon}{2}\right)^{k-t}\leq \left(1-\frac{1}{2^he^{h-1}}\right)^{\frac{k}{2}}.
$$
This concludes the proof.
\end{proof}

One way to interpret Theorem~\ref{Th:RandomField} is as saying that random codes cannot offer an asymptotic improvement upon the construction of Theorem~\ref{Th:OptimizedConstruction}.

\section{Open questions}\label{Sec:OpenQ}

We studied the trade-off between maximal recoverability and alphabet size in local codes. Most questions in this area remain open. The main challenge is to reduce the field in constructions of Theorems~\ref{Th:BasicConstruction} and~\ref{Th:OptimizedConstruction} or to prove that such a reduction is not possible.

\begin{enumerate}
\item We are particularly interested in the asymptotic setting of constant $r$ and $h$ and growing $k.$ In this setting can one get local MR codes over a field of size $O(k)$ or local MR codes inherently require a larger field than then their MDS counterparts?

\item In the setting of $h=O(1),$ $r=\Theta(k),$ and growing $k,$ can one get a lower bound of $\omega(k)$ for the field size of local MR codes?

\item While data-local and local codes present two important practically motivated code topologies, constructing MR codes over other topologies is also of interest. Below we sketch the general definitions of code topology and maximal recoverability.

    Assume there are two kinds of characters $\{x_i\}_{i\in [k]}$ and $\{\alpha_j\}_{j\in [t]}.$ Characters $\{x_i\}$ represent data symbols and characters $\{\alpha_j\}$ represent free coefficients. An $[n,k]$ systematic code topology is a collection of $n$ expressions $\{E_\ell\}_{\ell\in [n]}$ in $\{x_i\}$ and $\{\alpha_j\}.$ Such a collection includes all individual characters $x_1,\ldots,x_k.$ Every other expression has the form
    $$
    E_s=\sum_i L_{i,s}(\alpha_1,\ldots,\alpha_t)x_i,
    $$
    where $L_{i,s}$'s are arbitrary linear functions of $\{\alpha_j\}$ over a field $\mathbb{F}^\prime.$ Specifying code topology allows one to formally capture locality constraints that one wants to impose on the code. Fixing values of coefficients $\{\alpha_j\}$ in a field $\mathbb{F}$ extending $\mathbb{F}^\prime$ turns a code topology into a systematic $[n,k]$ code over~$\mathbb{F}.$

    We say that a sub-collection $S$ of expressions implies an expression $E_i$ if $E_i$ can obtained as a linear combination of expressions in $S,$ where the coefficients are rational functions in $\{\alpha_j\}.$ We say that an instantiation of a topology is maximally recoverable if every implication as above still holds after we instantiate $\{\alpha_j\}$'s. In other words, instantiating the corresponding rational functions does not cause a division by zero.
\end{enumerate}

\section*{Acknowledgements}
We would like to thank Swastik Kopparty and Raghu Meka for allowing us to include their Theorem~\ref{Th:RandomField} in the current paper.

\bibliographystyle{plain}
\bibliography{code-locality}

\end{document}